\documentclass[reqno]{amsart}
\usepackage{amssymb}

\newtheorem{theorem}{Theorem}
\newtheorem{proposition}[theorem]{Proposition}

\newtheorem{corollary}[theorem]{Corollary}

\theoremstyle{remark}

\numberwithin{equation}{section}

\begin{document}

\title[Boundary interactions for the difference Toda chain]
{Integrable boundary interactions for Ruijsenaars' difference Toda chain}

\author{J.F.  van Diejen}

\address{
Instituto de Matem\'atica y F\'{\i}sica, Universidad de Talca,
Casilla 747, Talca, Chile}

\email{diejen@inst-mat.utalca.cl}

\author{E. Emsiz}

\address{
Facultad de Matem\'aticas, Pontificia Universidad Cat\'olica de Chile,
Casilla 306, Correo 22, Santiago, Chile}
\email{eemsiz@mat.puc.cl}


\thanks{This work was supported in part by the {\em Fondo Nacional de Desarrollo
Cient\'{\i}fico y Tecnol\'ogico (FONDECYT)} Grants \# 1130226 and  \# 1141114.}

\date{May 2014}

\begin{abstract}
We endow Ruijsenaars' open difference Toda chain with a one-sided boundary interaction of
Askey-Wilson type and diagonalize the quantum Hamiltonian by means of deformed hyperoctahedral $q$-Whittaker functions that arise as a $t=0$ degeneration of the Macdonald-Koornwinder multivariate Askey-Wilson polynomials.
This immediately entails the quantum integrability,  the bispectral dual system, and the $n$-particle scattering operator for the chain in question.
\end{abstract}

\maketitle

\section{Introduction}\label{sec1}
It is well-known that the open and closed Toda chains may be viewed as limits of the hyperbolic and elliptic Calogero-Moser-Sutherland particle systems, respectively \cite{sut:introduction,rui:relativistic,ino:finite,rui:finite-dimensional}.
More general integrable open Toda chains
with boundary interactions involving potentials of Morse type \cite{kos:quantization,goo-wal:classical,skl:boundary}
and of P\"oschl-Teller type \cite{ino:finite,kuz-jor-chr:new}
are recovered similarly as degenerations of the Olshanetsky-Perelomov-Inozemtsev generalized Calogero-Moser-Sutherland systems with hyperoctahedral symmetry
\cite{ino:finite,osh:completely,shi:limit,ger-leb-obl:quantum}.
Moreover, such limiting relations turn out to persist at the level of the Ruijsenaars-Schneider particle systems and Ruijsenaars' difference (a.k.a. relativistic) Toda chains
\cite{rui:relativistic,rui:finite-dimensional,rui:systems,eti:whittaker,ger-leb-obl:q-deformed,hal-rui:kernel,bor-cor:macdonald} as well as their hyperoctahedral counterparts
\cite{die:difference,che:whittaker}.  Specifically, in the hyperoctahedral case one recovers in this manner generalizations of
Ruijsenaars' open relativistic Toda chain with boundary interactions that were studied at the level of classical mechanics in Refs.
\cite{sur:discrete,die:deformations,sur:problem} and at the level of quantum mechanics in Refs.
\cite{kuz-tsy:quantum,die:difference,eti:whittaker,sev:quantum,che:whittaker}.

In the present work we consider the Hamiltonian of such an open difference Toda chain
endowed with a one-sided four-parameter boundary interaction of Askey-Wilson type.
Upon diagonalizing the quantum Hamiltonian in question by means of deformed hyperoctahedral $q$-Whittaker functions that arise as a $t=0$ degeneration of the Macdonald-Koornwinder polynomials \cite{koo:askey-wilson,mac:affine}, the quantum integrability, the bispectral dual system, and the
$n$-particle scattering operator are deduced. For special values of the Askey-Wilson parameters, our chain amounts to a difference counterpart of the $D_n$-type and the $A_{n-1}$-type quantum Toda chains with one-sided boundary potentials of  P\"oschl-Teller and Morse type, respectively.

The presentation is structured as follows. After introducing our difference Toda chain in Section \ref{sec2} and defining the deformed hyperoctahedral $q$-Whittaker functions in Section \ref{sec3}, the diagonalization of the Hamiltonian is carried out in Section \ref{sec4} by identifying the corresponding eigenvalue equation with the 
$t\to 0$ degeneration of a well-known  Pieri formula for the Macdonald-Koornwinder polynomials \cite{die:properties,mac:affine}.
The quantum integrals and the bispectral dual system are then discussed in Sections \ref{sec5} and \ref{sec6}, respectively. In Section \ref{sec7} analogous results
for a difference counterpart of the quantum Toda chain
with one-sided boundary potentials of Morse type are obtained by letting one of the boundary parameters tend to zero (which corresponds to a transition
from Askey-Wilson polynomials to continuous dual $q$-Hahn polynomials
\cite{koe-les-swa:hypergeometric}). We close in Section \ref{sec8} with an explicit description of the $n$-particle scattering operator that relies on a stationary-phase analysis that was performed in Refs. \cite{rui:factorized,die:scattering}.
Some useful properties of the Macdonald-Koornwinder multivariate Askey-Wilson polynomials have been collected in a separate appendix at the end.

\section{Difference Toda chain with one-sided boundary interaction of Askey-Wilson type}\label{sec2}
Formally, the Hamiltonian of our difference Toda chain is given by the difference operator \cite{die:difference}:
\begin{subequations}
\begin{align}\label{Ha}
H:=&  T_1+ \sum_{j=2}^{n-1}  (1-q^{x_{j-1}-x_{j}})T_j \\
& +\sum_{j=1}^{n-2} (1-q^{x_j-x_{j+1}})T_j^{-1} +(1-q^{x_{n-1}-x_{n}})(1-q^{x_{n-1}+x_{n}})T_{n-1}^{-1}\nonumber \\
& +w_+(x_n)(1-q^{x_{n-1}-x_{n}})T_{n} +
w_-(x_n)(1-q^{x_{n-1}+x_{n}})T_{n}^{-1} + U(x_{n-1},x_n), \nonumber
\end{align}
where
\begin{align}
w_+ (x) &:=\frac{ \prod_{0\leq r \leq 3} (1-{t}_r q^{x})}{(1-q^{2x})(1-q^{2x+ 1})},\quad
w_- (x) :=\frac{ \prod_{0\leq r \leq 3} (1-{t}_r^{-1} q^{x})}{(1-q^{2x})(1-q^{2x- 1})},\\
U(x,y) &: =  \sum_{\epsilon \in \{ 1,-1\}}  \frac{c_\epsilon(1-\epsilon q^{x+1/2})}{(1-\epsilon q^{y-1/2})(1-\epsilon q^{-y-1/2})}, 
\end{align}
with
\begin{equation}\label{Hd}
c_\epsilon := \frac{1}{2\sqrt{q^{-1}{t}_0{t}_1{t}_2{t}_3}}\prod_{0\leq r \leq 3} (1-\epsilon q^{-1/2}{t}_r) ,
\end{equation}
\end{subequations}
and $T_j$ ($j=1,\ldots ,n$) acts on functions $f:\mathbb{R}^n\to\mathbb{C}$ by a
unit translation of the $j$th position variable $$(T_jf)(x_1,\ldots,x_n)=f(x_1,\ldots,x_{j-1},x_j+1,x_{j+1},\ldots,x_n).$$
Here $q$  denotes a scale parameter and the parameters ${t}_r$ ($r=0,\ldots ,3$)
play the role of coupling parameters for the  boundary interaction of Askey-Wilson type.
Upon setting ${t}_2=-{t}_3=q^{1/2}$, the additive potential term $U(x_{n-1},x_n)$ in $H$ \eqref{Ha}--\eqref{Hd} vanishes. The above Toda chain amounts in this case to a difference analog of the previously studied $D_n$-type quantum Toda chain with P\"oschl-Teller boundary potential
\cite{ino:finite,kuz-jor-chr:new,osh:completely,ger-leb-obl:quantum}. If we additionally set ${t}_0=-{t}_1=1$, then $w_+(x)=w_-(x)=1$ and we formally recover a $D_n$-type analog of Ruijsenaars' difference Toda chain
 \cite{kuz-tsy:quantum,eti:whittaker,sev:quantum,che:whittaker} that was introduced at the level of classical mechanics by Suris \cite{sur:discrete}.

\section{Deformed hyperoctahedral $q$-Whittaker functions}\label{sec3}
Let  $\Lambda$ denote the cone of integer partitions $\lambda=(\lambda_1,\ldots ,\lambda_n)$  with
decreasingly ordered parts $\lambda_1\geq\cdots\geq\lambda_n\geq 0$, and let
$W$ be the hyperoctahedral group formed by the semi-direct product of the symmetric group $S_n$ and the
$n$-fold product of the cyclic group $\mathbb{Z}_2\cong \{1,-1\}$. 
Elements $w=(\sigma,\epsilon)\in W$ act naturally on $\xi=(\xi_1,\ldots\xi_n)\in \mathbb{R}^n$ via $w\xi:=(\epsilon_1 \xi_{\sigma_1},\ldots ,\epsilon_n \xi_{\sigma_n})$ (with $\sigma\in S_n$ and $\epsilon_j\in \{1,-1\}$ for $j=1,\ldots ,n$).
A standard basis for
the algebra of $W$-invariant trigonometric polynomials on the torus $\mathbb{T}=\mathbb{R}^n/(2\pi\mathbb{Z}^n)$ is given by
the
hyperoctahedral monomial symmetric functions
\begin{equation}
m_\lambda(\xi):=\sum_{\mu\in W\lambda} e^{i\langle \mu ,\xi\rangle},\qquad \lambda\in\Lambda,
\end{equation}
where the summation is meant over the orbit of $\lambda$ with respect to the action of $W$ and the bracket $\langle\cdot ,\cdot \rangle$ refers to the usual inner product on $\mathbb{R}^n$ (so $\langle \mu ,\xi\rangle =\mu_1\xi_1+\cdots +\mu_n\xi_n$). This monomial basis inherits a natural partial order from the
hyperoctahedral dominance ordering of the partitions:
\begin{equation}\label{porder}
\forall \mu,\lambda\in\Lambda:\quad \mu\leq \lambda \ \text{iff}\
\sum_{1\leq j\leq k} \mu_j \leq \sum_{1\leq j\leq k} \lambda_j \quad\text{for}\quad k=1,\ldots ,n.
\end{equation}

By definition, the basis of deformed hyperoctahedral $q$-Whittaker functions
$p_\lambda(\xi)$, $\lambda\in\Lambda$  is given by
the polynomials of the form
\begin{subequations}
\begin{equation}\label{qW1}
p_\lambda(\xi) = m_\lambda(\xi) +\sum_{\substack{\mu\in\Lambda\\\text{with}\, \mu <\lambda}}
c_{\lambda ,\mu}m_\mu(\xi)\qquad (c_{\lambda ,\mu}\in\mathbb{C})
\end{equation}
such that
\begin{equation}\label{qW2}
\langle p_\lambda ,m_\mu\rangle_{\hat{\Delta}} = 0\quad \text{if}\ \mu <\lambda ,
\end{equation}
\end{subequations}
where the  inner product 
\begin{subequations}
\begin{equation}\label{ip}
\langle \hat{f},\hat{g}\rangle_{\hat{\Delta}}:=\int_{\mathbb{A}}\hat{f}(\xi)\overline{\hat{g}(\xi)}\hat{\Delta}(\xi)\text{d}\xi
\qquad (\hat{f},\hat{g}\in L^2(\mathbb{A},\hat{\Delta} (\xi) \text{d}\xi))
\end{equation}
is determined by the weight function
\begin{equation}\label{plancherel}
\hat{\Delta} (\xi ):=\frac{1}{(2\pi)^n}
\prod_{1\leq j<k\leq n}  \left| (e^{i(\xi_j+\xi_k)},e^{i(\xi_j-\xi_k)})_\infty \right|^2
\prod_{1\leq j\leq n}\left|
\frac{(e^{2i\xi_j})_\infty}{\prod_{0\leq r\leq 3} (\hat{t}_r e^{i\xi_j})_\infty}   \right|^2
\end{equation}
\end{subequations}
supported on the hyperoctahedral Weyl alcove
\begin{equation}\label{alcove}
\mathbb{A}:=\{ (\xi_1,\xi_2,\ldots,\xi_n)\in\mathbb{R}^n\mid \pi>\xi_1>\xi_2>\cdots >\xi_n>0\} .
\end{equation}
Here $(x)_m:=\prod_{l =0}^m (1-xq^l)$ and $(x_1,\ldots, x_l)_m:=(x_1)_m\cdots (x_l)_m$ refer to standard notations for the $q$-Pochhammer symbols, and  it  is assumed that
\begin{equation}\label{wp}
q\in (0,1)\quad\text{and}\quad \hat{t}_r\in (-1,1)\setminus \{ 0\} \quad (r=0,\ldots ,3).
\end{equation}
These deformed hyperoctahedral  $q$-Whittaker functions $p_\lambda(\xi)$, $\lambda\in\Lambda$ amount to a $t\to 0$ degeneration of the more general Macdonald-Koorwinder multivariate Askey-Wilson polynomials introduced in Ref. \cite{koo:askey-wilson} (cf. Appendix \ref{appA} below).

\section{Diagonalization}\label{sec4}
It is known that the eigenfunctions of Ruijsenaars' open difference Toda chain consist of $A_{n-1}$-type  $q$-Whittaker functions given by a $t\to 0$ limit of the Macdonald symmetric functions \cite{ger-leb-obl:q-deformed}.
In this section our aim is to show that an analogous result holds for the chain with Askey-Wilson type boundary interactions from Section \ref{sec2}, upon employing the deformed hyperoctahedral $q$-Whittaker functions from Section \ref{sec3}.
To this end it is convenient to reparametrize the boundary parameters of the Toda chain in terms of the $q$-Whittaker deformation parameters \eqref{wp}  via
\begin{equation}\label{tp1}
t_0=\sqrt{q^{-1}\hat{t}_0\hat{t}_1\hat{t}_2\hat{t}_3},\qquad
t_r=\hat{t}_r\hat{t}_0/t_0\quad (r=1,2,3),
\end{equation}
assuming (from now onwards) the additional positivity constraints
\begin{equation}\label{tp2}
\hat{t}_0>0\quad \text{and}\quad \hat{t}_0\hat{t}_1\hat{t}_2\hat{t}_3>0.
\end{equation}

Let $\rho_0+\Lambda:=\{\rho_0+\lambda\mid\lambda\in\Lambda\}$  with $$\rho_0:=(\log_q({t}_0),\ldots,\log_q({t}_0))\in\mathbb{R}^n.$$
We write
$\ell^2(\rho_0+\Lambda,{\Delta})$ for the Hilbert space of lattice functions ${f}:(\rho_0+\Lambda)\to\mathbb{C}$ determined by the inner product
\begin{subequations}
\begin{equation}
\langle {f},{g}\rangle_{{\Delta}}:=
\sum_{\lambda\in\Lambda}{f}(\rho_0+\lambda)\overline{{g}(\rho_0+\lambda)}{\Delta}_\lambda
\qquad ({f},{g}\in \ell^2(\rho_0+\lambda_n,{\Delta})),
\end{equation}
where
\begin{equation}
{\Delta}_\lambda:= \frac{{\Delta}_0}{(q{t}_0^2)_{\lambda_{n-1}+\lambda_n}}
\Bigl(\frac{1-{t}_0^2q^{2\lambda_n}}{1-{t}_0^2}\Bigr)
\prod_{0\leq r\leq 3} \frac{({t}_0{t}_r)_{\lambda_n}}{(q{t}_0{t}_r^{-1})_{\lambda_n}}
\prod_{1\leq j<n}\frac{1}{(q)_{\lambda_j-\lambda_{j+1}}} 
\end{equation}
and
\begin{equation}
{\Delta}_0:=(q)_\infty \prod_{0\leq r<s\leq 3} (\hat{t}_r\hat{t}_s)_\infty 
=(q)_\infty \prod_{1\leq r\leq 3} (t_0t_r,qt_0t_r^{-1})_\infty .
\end{equation}
\end{subequations}
From the limiting behavior for $t\to 0$ of the orthogonality relations satisfied by the normalized
Macdonald-Koornwinder polynomials \eqref{mk-ortho-a}--\eqref{mk-ortho-d}, it is immediate that the wave function
\begin{equation}\label{dqwf}
\psi_\xi(\rho_0+\lambda):=\frac{({t}_0^2)_{2\lambda_n}}{\prod_{0\leq r\leq 3}( {t}_0{t}_r)_{\lambda_n}}p_\lambda(\xi)\qquad (\lambda\in\Lambda,\, \xi\in\mathbb{A})
\end{equation}
satisfies the following orthogonality with respect to the spectral variable $\xi$:
\begin{equation}\label{qw-orthogonality}
\int_{\mathbb{A}} \psi (\rho_0+\lambda)\overline{ \psi (\rho_0+\mu)}\hat{\Delta}(\xi)\text{d}\xi =
\begin{cases}
{\Delta}_\lambda^{-1} &\text{if}\ \lambda =\mu ,\\
0&\text{otherwise}.
\end{cases}
\end{equation}
In other words, the corresponding Fourier transform
 $\boldsymbol{F}: \ell^2(\rho_0+\Lambda,{\Delta})\to L^2(\mathbb{A},\hat{\Delta}\text{d}\xi)$
given by
 \begin{subequations}
\begin{equation}\label{ft1}
(\boldsymbol{F}{f})(\xi):= \langle {f},\psi_\xi \rangle_{{\Delta}}=\sum_{\lambda\in\Lambda}f(\rho_0+\lambda)
\overline{\psi_\xi (\rho_0+\lambda)}{\Delta}_\lambda
\end{equation}
(${f}\in \ell^2(\rho_0+\Lambda,{\Delta})$) constitutes a Hilbert space isomorphism
with an inversion formula of the form
\begin{equation}\label{ft2}
(\boldsymbol{F}^{-1}\hat{f})(\rho_0+\lambda) = \langle \hat{f},\overline{\psi ( \rho_0+\lambda)}\rangle_{\hat{\Delta}}=
\int_{\mathbb{A}} \hat{f} (\xi) \psi_\xi (\rho_0+\lambda)\hat{\Delta}(\xi)\text{d}\xi
\end{equation}
\end{subequations}
($\hat{f}\in L^2(\mathbb{A},\hat{\Delta}\text{d}\xi)$). We will refer to $\boldsymbol{F}$ \eqref{ft1}, \eqref{ft2} as the deformed hyperoctahedral $q$-Whittaker transform.

The formal Hamiltonian $H$ \eqref{Ha}--\eqref{Hd} restricts to a well-defined discrete difference operator in the space of complex functions on the lattice $\rho_0+\Lambda$. Indeed, when $t_0\not\in \{ 1 , q^{1/2}\}$ it is manifest that for $x=(x_1,\ldots ,x_n)$ at these lattice points we stay away from the poles in the coefficients of $H$ stemming from the denominators of  $w_\pm(x_n)$ and $U(x_{n-1},x_n)$ and, moreover, that
for any $f:\mathbb{R}^n\to\mathbb{C}$ and any $\lambda\in\Lambda$ the value of $(Hf)(\rho_0+\lambda)$ 
depends only on evaluations
of $f$ at points of $\rho_0+\Lambda$ (due to the vanishing of
$(1-q^{\lambda_j-\lambda_{j+1}})$ at $\lambda_j=\lambda_{j+1}$ ($1\leq j<n$) and the vanishing of $w_-(\log_q(t_0)+\lambda_n)$ at $\lambda_n=0$):
\begin{eqnarray}\label{Haction}
\lefteqn{(Hf)(\rho_0+\lambda)=}&& \\
&& \sum_{\substack{1\leq j\leq n\\ \lambda +e_j\in\Lambda}} v_j^+(\lambda) f (\rho_0+\lambda +e_j)+ \sum_{\substack{1\leq j\leq n\\ \lambda -e_j\in\Lambda}}v_j^-(\lambda)f (\rho_0+\lambda -e_j) 
+u(\lambda )f (\rho_0+\lambda) ,\nonumber
\end{eqnarray}
where
\begin{align*}
v_j^+(\lambda) =&
(1-q^{\lambda_{j-1}-\lambda_j})\left(\frac{ \prod_{0\leq r \leq 3} (1-{t}_rt_0 q^{\lambda_n})}{(1-t_0^2q^{2\lambda_n})(1-t_0^2q^{2\lambda_n+ 1})}\right)^{\delta_{n-j}} , \\
v_j^-(\lambda) =&
(1-q^{\lambda_{j}-\lambda_{j+1}})(1-t_0^2 q^{\lambda_{n-1}+\lambda_{n}})^{\delta_{n-j}+\delta_{n-1-j}}\\
&\times
\left(\frac{ \prod_{0\leq r \leq 3} (1-{t}_r^{-1}t_0 q^{\lambda_n})}{(1-t_0^2q^{2\lambda_n})(1-t_0^2q^{2\lambda_n- 1})}\right)^{\delta_{n-j}} ,\\
u(\lambda )=&  \sum_{\epsilon \in \{ 1,-1\}}  \frac{c_\epsilon(1-\epsilon t_0 q^{\lambda_{n-1}+1/2})}{(1-\epsilon t_0q^{\lambda_n-1/2})(1-\epsilon t_0^{-1} q^{-\lambda_n-1/2})} ,
\end{align*}
with  $c_\epsilon$ taken from \eqref{Hd}. 
Here $\delta_k:=1 $ if $k=0$ and $\delta_k:=0$ otherwise,
the vectors $e_1,\ldots ,e_n$ denote the standard unit basis of $\mathbb{R}^n$, and
 $\lambda_{0}:=+\infty$, $\lambda_{n+1}:=-\infty$ by convention (so $(1-q^{\lambda_0-\lambda_1})= (1-q^{\lambda_n-\lambda_{n+1}})\equiv 1$).
The action of $H$ on lattice functions  in Eq. \eqref{Haction}  extends continuously from  $t_0\not\in \{ 1 , q^{1/2}\}$ to the full parameter domain determined by Eqs. \eqref{tp1}, \eqref{tp2} and \eqref{wp}.

Our main result implements the Hamiltonian under consideration as a self-adjoint operator in the Hilbert space
$ \ell^2(\rho_0+\Lambda,{\Delta})$ and provides its spectral decomposition with the aid of the
deformed hyperoctahedral $q$-Whittaker transform.

 \begin{theorem}[Diagonalization]\label{diagonal:thm}
(i). For boundary parameters $t_r$  \eqref{tp1} determined by the $q$-Whittaker deformation parameters $\hat{t}_r$ \eqref{wp},  \eqref{tp2}, the action of the
 difference Toda Hamiltonian $H$ \eqref{Ha}--\eqref{Hd}  given by Eq. \eqref{Haction} 
 constitutes a bounded self-adjoint operator in the Hilbert space $\ell^2(\rho_0+\Lambda,\Delta)$  
 with
purely absolutely continuous spectrum.
(ii). The operator in question is diagonalized by
 the deformed hyperoctahedral $q$-Whittaker transform $\boldsymbol{F}$ \eqref{ft1}, \eqref{ft2}:
 \begin{subequations}
 \begin{equation}\label{Hdiagonal}
H=\boldsymbol{F}^{-1}  \circ \hat{{E}} \circ\boldsymbol{F},
 \end{equation}
 where $\hat{{E}}$ denotes the bounded real multiplication operator acting on
$\hat{f}\in L^2(\mathbb{A},\hat{\Delta}\text{d}\xi)$ via
\begin{equation}\label{E}
(\hat{{E}}\hat{f})(\xi):=\hat{E}(\xi) \hat{f}(\xi)\quad\text{with}\quad \hat{E}(\xi):=2\sum_{1\leq j\leq n}\cos(\xi_j) .
\end{equation}
\end{subequations}
\end{theorem}

\begin{proof}
The first part  of the theorem is immediate form the second part. To prove the second part
it suffices to verify that the deformed hyperoctahedral $q$-Whittaker kernel $\psi_\xi$ satisfies
the eigenvalue equation
$H\psi_\xi=\hat{E}(\xi)\psi_\xi$, or more explicitly that:
\begin{align*}
\sum_{\substack{1\leq j\leq n\\ \lambda +e_j\in\Lambda}} v_j^+(\lambda) \psi_\xi (\rho_0+\lambda +e_j)+ \sum_{\substack{1\leq j\leq n\\ \lambda -e_j\in\Lambda}}v_j^-(\lambda) \psi_\xi (\rho_0+\lambda -e_j) & \\
+u(\lambda )\psi_\xi (\rho_0+\lambda) 
=\hat{E}(\xi)\psi_\xi (\rho_0+\lambda) .& \nonumber
\end{align*}
This eigenvalue equation follows from the Pieri formula for the Macdonald-Koornwinder polynomials \eqref{pieri} in the limit $t\to 0$. Indeed, it is clear that in the Pieri formula
$\lim_{t\to 0} \mathbf{P}_\lambda (\xi)=\psi_\lambda (\rho_0+\lambda )$,
$\lim_{t\to 0} \hat{\tau}_jV^+_j(\lambda)=v_j^+(\lambda)$, $\lim_{t\to 0} \hat{\tau}_j^{-1}V^-_j(\lambda)=v_j^-(\lambda)$, and one also has that
\begin{equation*}\label{limit}
 \lim_{t\to 0} \Bigl( \sum_{j=1}^n (\hat{\tau}_j+\hat{\tau}_j^{-1})-
\sum_{\substack{1\leq j\leq n\\ \lambda +e_j\in\Lambda}}
V_j^+(\lambda) 
-\sum_{\substack{1\leq j\leq n\\ \lambda -e_j\in\Lambda}}
V_j^-(\lambda)  \Bigr) = u(\lambda) .
\end{equation*}
This last limit formula is not evident
but can be deduced from the following rational identity in $q^{x_1},\ldots, q^{x_n}$:
\begin{align*}
&\sum_{j=1}^n\Bigl( \hat{\tau}_j^{-1}-  \hat{\tau}_1^{-1}w_+(x_j) \prod_{\substack{1\leq k\leq n\\k\neq j}}
\frac{1-tq^{x_j+x_k}}{1-q^{x_j+x_k}}
\frac{1-tq^{x_j-x_k}}{1-q^{x_j-x_k}} \Bigr) +\\
&\sum_{j=1}^n  \Bigl( \hat{\tau}_j- \hat{\tau}_1w_-(x_j) 
\prod_{\substack{1\leq k\leq n\\k\neq j}}
\frac{1-t^{-1}q^{x_j+x_k}}{1-q^{x_j+x_k}}
\frac{1-t^{-1}q^{x_j-x_k}}{1-q^{x_j-x_k}} \Bigr) \nonumber \\
&=C_t\sum_{\epsilon \in \{ 1,-1\}}  \prod_{0\leq r\leq 3}(1-\epsilon t_rq^{-1/2}) \Bigl( 1-\prod_{j=1}^n
\frac{1-\epsilon t q^{x_j-1/2}}{1-\epsilon q^{x_j-1/2}}
\frac{1-\epsilon t^{-1} q^{x_j+1/2}}{1-\epsilon q^{x_j+1/2}} \Bigr) , \nonumber
\end{align*}
where $C_t=-\frac{1}{2}t \hat{t}_0^{-1}(1-t)^{-1}(1-q^{-1}t)^{-1}$, upon replacing
$q^{x_j}$ by $\tau_jq^{\lambda_j}$ ($j=1,\ldots ,n$) and
performing the limit $t\to 0$. To infer the rational identity itself, one exploits
the hyperoctahedral symmetry in the variables $x_1,\ldots x_n$ and checks that---as a function of $x_j$ (with the remaining variables fixed in a generic configuration)---the residues at the (simple) poles on both sides coincide. Hence, the difference of both rational expressions amounts to a $W$-invariant Laurent polynomial in $q^{x_1},\ldots ,q^{x_n}$.
The Laurent polynomial in question must actually vanish, as the rational expressions under consideration tend to $0$
for $x_j=(n+1-j)c$ in the limit $c\to +\infty$.
\end{proof}

\section{Integrability}\label{sec5}
The quantum integrability of the difference Toda Hamiltonian $H$ \eqref{Ha}--\eqref{Hd}
is an immediate consequence of
the diagonalization in Theorem \ref{diagonal:thm}.
In effect, a complete system of commuting
 quantum integrals in the Hilbert space $ \ell^2(\rho_0+\Lambda,{\Delta})$
is given by the bounded self-adjoint operators
\begin{equation}\label{integrals}
H_l :=\boldsymbol{F}^{-1}  \circ \hat{E}_l \circ\boldsymbol{F},\qquad l=1,\ldots, n,
 \end{equation}
where $\hat{E}_l : L^2(\mathbb{A},\hat{\Delta}\text{d}\xi)\to L^2(\mathbb{A},\hat{\Delta}\text{d}\xi)$ 
denotes the real multiplication operator by $\hat{E}_l(\xi):=m_{\omega_l}(\xi)$ with $\omega_l:=e_1+\cdots +e_l$ (so $H_1=H$). The operator $H_l$ \eqref{integrals} acts on $f\in \ell^2(\rho_0+\Lambda,{\Delta})$ as a difference operator of the form
 \begin{subequations}
 \begin{equation}
( H_l f)(\rho_0+\lambda)= \sum_{\substack{J\subset\{1,\ldots ,n\}  , 0\leq |J|\leq l \\ \epsilon_j\in \{ 1, -1\}, j\in J ; \lambda + e_{\epsilon J}\in\Lambda   }     } C^{(l )}_{\epsilon J} (\lambda) f(\rho_0+\lambda +e_{\epsilon J}),
  \end{equation}
 where $e_{\epsilon J}:=\sum_{j\in J} \epsilon_j e_j$, $|J|$ denotes the cardinality of $J\subset \{ 1,\ldots ,n\}$,
 and the coefficients
 \begin{equation}
 C^{(l)}_{\epsilon J} (\lambda)=\lim_{t\to 0} C^{(l)}_{\epsilon J,t} (\lambda) 
 \end{equation}
 arise as $t\to 0$ limits of the expansion coefficients in the corresponding Pieri formula for the normalized Macdonald-Koornwinder polynomials $\mathbf{P}_\lambda (\xi)$ \eqref{n-pol1}, \eqref{n-pol2} (cf.
 \cite[\text{Sec.}~6]{die:properties}):
 \begin{equation}
 \hat{E}_l(\xi)\mathbf{P}_\lambda(\xi)=\sum_{\substack{J\subset\{1,\ldots ,n\}  , 0\leq |J|\leq l \\ \epsilon_j\in \{ 1, -1\}, j\in J ; \lambda + e_{\epsilon J}\in\Lambda   }     }  C^{(l)}_{\epsilon J,t} (\lambda) \mathbf{P}_{\lambda+e_{\epsilon J}}(\xi) .
 \end{equation}
 \end{subequations}
  Notice in this connection that the Pieri expansion coefficients
  $$C^{(l)}_{\epsilon J,t} (\lambda)=\boldsymbol{\Delta}_{\lambda+e_{\epsilon J}}\int_{\mathbb{A}}  \hat{E}_l(\xi)\mathbf{P}_\lambda(\xi) \overline{ \mathbf{P}_{\lambda+e_{\epsilon J}}(\xi)}\boldsymbol{\hat{\Delta}}(\xi)\text{d}\xi $$
 are continuous at $t=0$, because the Macdonald-Koornwinder  weight function $\boldsymbol{\hat{\Delta}}(\xi)$ and (thus) the polynomials $\mathbf{P}_\lambda (\xi)$, $\lambda\in\Lambda$ are continuous at this parameter value
  (cf. Appendix \ref{appA}). 
 
 In practice it turns out to be very tedious to compute the $t\to 0$ limiting coefficients
 $C^{(l)}_{\epsilon J} (\lambda) $  explicitly  with  the aid of the known explicit Pieri formulas for the
 Macdonald-Koornwinder polynomials  in \cite[\text{Sec.}~6]{die:properties} beyond $l =1$. 
For a particular second quantum integral belonging to the commutative algebra generated by $H_1,\ldots ,H_n$, however,  the required computation results to be surprisingly straightforward.
More specifically:  from the $t\to 0$ limiting behavior of the $r=n$ (top)
Pieri formula for the Macdonald-Koornwinder polynomials in Theorem 6.1 of \cite{die:properties}, one readily deduces that the action on $f\in\ell^2(\rho_0+\Lambda,\Delta)$ of the operator $H_Q:= \boldsymbol{F}^{-1}  \circ \hat{Q} \circ\boldsymbol{F}$, where $\hat{Q}$
 refers to the self-adjoint multiplication operator in
$L^2(\mathbb{A},\hat{\Delta}\text{d}\xi)$ by $$\hat{Q}(\xi):=\prod_{j=1}^n (2\cos (\xi_j)-\hat{t}_0-\hat{t}_0^{-1}),$$ is given explicitly by 
\begin{eqnarray}\label{HQ}
\lefteqn{ (H_Qf)(\rho_0+\lambda ) = } && \\
 && \sum_{\substack{J_+\cup J_-\cup K_+\cup K_-=\{1,\ldots,n\}   \\ |J_+|+| J_- |+ | K_+ |+| K_- |=n \\ \lambda + e_{J_+}-e_{J_-}\in\Lambda }     } 
 u_{K_+,K_-}(\lambda) v_{J_+,J_-}(\lambda ) f(\rho_0+\lambda+e_{J_+}-e_{J_-}) , \nonumber
 \end{eqnarray}
 with
 \begin{align*}
&  v_{J_+,J_-}(\lambda )  =
 \prod_{\substack{j\in J_+ \\ j-1\not\in J_+}} (1-q^{\lambda_{j-1}-\lambda_j})
  \prod_{\substack{j\in J_- \\ j+1\not\in J_-}} (1-q^{\lambda_{j}-\lambda_{j+1}-\delta_{J_+}(j+1)})  \\
 & \times (1-t_0^2q^{\lambda_{n-1}+\lambda_n})^{\delta_{J_+^c}(n-1)\delta_{J_+^c}(n)-\delta_{J_+^c\cap J_-^c}(n-1)\delta_{J_+^c\cap J_-^c}(n)} \\
&\times (1-t_0^2q^{\lambda_{n-1}+\lambda_n-1})^{\delta_{J_-}(n-1)\delta_{J_-}(n)}  \
w_+(\lambda_n)^{\delta_{J_+}(n)}w_-(\lambda_n)^{\delta_{J_-}(n)}
 \end{align*}
 and
 \begin{align*}
u_{K_+,K_-}(\lambda)  &=(-\hat{t}_0)^{|K_-|-|K_+|}  \prod_{\substack{k\in K_+\\ k-1\in K_-}} (1-q^{\lambda_{k-1}-\lambda_{k}}) 
  \prod_{\substack{k\in K_+ \\ k+1\in K_-}} (1-q^{\lambda_{k}-\lambda_{k+1}+1})\\
  &\times (1-t_0^2q^{\lambda_{n-1}+\lambda_n+1})^{\delta_{K_+}(n-1)\delta_{K_+}(n)}
  (1-t_0^2q^{\lambda_{n-1}+\lambda_n})^{\delta_{K_-}(n-1)\delta_{K_-}(n)} \\
  &\times w_+(\lambda_n)^{\delta_{K_+}(n)}w_-(\lambda_n)^{\delta_{K_-}(n)} .
  \end{align*}
Here $\delta_J:\{ 1,\ldots ,n\} \to \{ 0, 1\}$ denotes the characteristic function of $J\subset \{1,\ldots ,n\}$ and $J^c=\{ 1,\ldots ,n\}\setminus J$.
 
 \begin{corollary}
 The difference Toda Hamiltonians $H$ \eqref{Haction} and $H_Q$ \eqref{HQ}  are bounded, self-adjoint,
 commuting operators in $ \ell^2(\rho_0+\Lambda,{\Delta})$ for which the deformed hyperoctahedral $q$-Whittaker functions $\psi_\xi$ \eqref{dqwf} constitute a complete system of (generalized) joint eigenfunctions corresponding to the eigenvalues $\hat{E}(\xi)$ and $\hat{Q}(\xi)$, respectively.
 \end{corollary}

\section{Bispectral dual system}\label{sec6}
For $t\to 0$ the Macdonald-Koornwinder $q$-difference equation \eqref{qde}
amounts to the following eigenvalue equation satisfied by the deformed hyperoctahedral $q$-Whittaker functions:
\begin{equation}
\hat{H}p_\lambda = (q^{-\lambda_1}-1) p_\lambda  \qquad (\lambda\in\Lambda),
\end{equation}
with
\begin{subequations}
\begin{equation}\label{hd1}
\hat{H}=\sum_{j=1}^n \left( \hat{v}_j(\xi)(\hat{T}_{j,q}-1)+ \hat{v}_{j}(-\xi)\hat{T}_{j,q}^{-1}-1)\right) ,
\end{equation}
and
\begin{equation}\label{hd2}
\hat{v}_{ j}(\xi)=
\frac{\prod_{0\leq r\leq 3} (1-\hat{t}_re^{i\xi_j}) }{(1-e^{2i\xi_j})  (1-q e^{2i\xi_j}) }
\prod_{\substack{1\leq k\leq n\\k\neq j}} (1-e^{i(\xi_j+\xi_k)})^{-1}(1-e^{i(\xi_j-\xi_k)})^{-1}, 
\end{equation}
\end{subequations}
where $\hat{T}_{j,q}$ acts on trigonometric (Laurent) polynomials $\hat{p}(e^{i\xi_1},\ldots ,e^{i\xi_n})$  by a $q$-shift of the $j$th variable:
\begin{equation*}
(\hat{T}_{j,q}\hat{p})(e^{i\xi_1},\ldots ,e^{i\xi_n}):=\hat{p}(e^{i\xi_1},\ldots,e^{i\xi_{j-1}},qe^{i\xi_j},e^{i\xi_{j+1}},\ldots,e^{i\xi_n}) .
\end{equation*}
The following proposition is now immediate.

\begin{proposition}[Bispectral Dual Hamiltonian]\label{dual:prp}
The $t=0$ Macdonald-Koornwinder $q$-difference operator $\hat{H}$  \eqref{hd1},\eqref{hd2} constitutes a  nonnegative unbounded self-adjoint operator with purely discrete spectrum in
$L^2(\mathbb{A},\hat{\Delta}\text{d}\xi)$ that is
diagonalized by the (inverse) deformed hyperoctahedral $q$-Whittaker transform $\boldsymbol{F}$ \eqref{ft1}, \eqref{ft2}:
\begin{subequations}
 \begin{equation}\label{Hddiagonal}
\hat{H}=\boldsymbol{F}  \circ {{E}} \circ\boldsymbol{F}^{-1},
 \end{equation}
where $E$ denotes the self-adjoint multiplication operator in $ \ell^2(\rho_0+\Lambda,\Delta)$ of the form
\begin{equation}\label{Ed}
(E f)(\rho_0+\lambda):=(q^{-\lambda_1}-1) f(\rho_0+\lambda) \qquad (\lambda\in\Lambda)
\end{equation}
\end{subequations}
 (for $f \in \ell^2(\rho_0+\Lambda,\Delta)$ with $\langle Ef,Ef\rangle_{\Delta}<\infty$).
\end{proposition}

One learns from Theorem \ref{diagonal:thm} and Proposition \ref{dual:prp} that the eigenfunction transforms diagonalizing the difference Toda Hamiltonian $H$ \eqref{Haction}
and  the  $t=0$ Macdonald-Koornwinder difference operator  $\hat{H}$ \eqref{hd1},\eqref{hd2}  are inverses of each other. This fact encodes the bispectral duality of the operators under consideration in the sense of Duistermaat and Gr\"unbaum \cite{dui-gru:differential,gru:bispectral}: the kernel function $\psi_\xi (\rho_0+\lambda)$ of the deformed hyperoctahedral $q$-Whittaker transform $\boldsymbol{F}$ \eqref{ft1}, \eqref{ft2} simultaneously solves the corresponding eigenvalue equations for $H$ and $\hat{H}$ in the discrete variable $\rho_0+\lambda$ and the spectral variable $\xi$, respectively.

Explicit commuting quantum integrals for the dual Hamiltonian $\hat{H}$  \eqref{hd1},\eqref{hd2} are obtained as a $t\to 0$ degeneration of the commuting difference operators in  \cite[\text{Thm.}~5.1]{die:properties}:
\begin{equation}\label{dual-int}
\hat{H}_l= 
\sum_{\substack{ J\subset \{ 1,\ldots ,n\},\, 0\leq |J|\leq l \\ \epsilon_j\in\{ 1,-1\},j\in J}} \hat{U}_{J^c,l -|J|}\hat{V}_{\epsilon J} \hat{T}_{\epsilon J,q},\qquad l=1,\ldots ,n,
\end{equation}
with $\hat{T}_{\epsilon J,q}:=\prod_{j\in J} \hat{T}_{j,q}^{\epsilon_j}$ and
\begin{align*}
\hat{V}_{\epsilon J} := &
\prod_{j\in J}  \frac{\prod_{0\leq r\leq 3} (1-\hat{t}_re^{i\epsilon_j\xi_j}) }{(1-e^{2i\epsilon_j\xi_j})  (1-q e^{2i\epsilon_j\xi_j}) } \prod_{\substack{j\in J \\k\not\in J}} (1-e^{i(\epsilon_j\xi_j+\xi_k)})^{-1}(1-e^{i(\epsilon_j \xi_j-\xi_k)})^{-1}
 \\
& \times \prod_{\substack{j,k\in J \\ j<k}} (1-e^{i(\epsilon_j\xi_j+\epsilon_k \xi_k)})^{-1}(1-qe^{i(\epsilon_j \xi_j+\epsilon_k\xi_k)})^{-1} ,
\end{align*}
\begin{align*}
\hat{U}_{K,p} :=  (-1)^p 
\sum_{\substack{ I \subset K,\,  | I | = p\\ \epsilon_j\in\{ 1,-1\},j\in I}} 
& \bigg( 
 \prod_{j\in I}  \frac{\prod_{0\leq r\leq 3} (1-\hat{t}_re^{i\epsilon_j\xi_j}) }{(1-e^{2i\epsilon_j\xi_j})  (1-q e^{2i\epsilon_j\xi_j}) } \\
&\times   \prod_{\substack{j\in I \\k\in K\setminus I}} (1-e^{i(\epsilon_j\xi_j+\xi_k)})^{-1}(1-e^{i(\epsilon_j \xi_j-\xi_k)})^{-1} \\
 & \times \prod_{\substack{j,k\in I \\ j<k}} (1-e^{i(\epsilon_j\xi_j+\epsilon_k \xi_k)})^{-1}(1-q^{-1}e^{-i(\epsilon_j \xi_j+\epsilon_k\xi_k)})^{-1} \bigg) 
\end{align*}
(so $\hat{H}_1=\hat{H}$). 
The diagonalization in Proposition \ref{dual:prp} now generalizes to the complete system of commuting quantum integrals $\hat{H}_1,\ldots ,\hat{H}_n$ as follows.
\begin{theorem}[Bispectral Dual System]\label{dual:thm} Let $E_l$ ($1\leq l \leq n$) denote the self-adjoint multiplication operator in $ \ell^2(\rho_0+\Lambda,\Delta)$ given by
\begin{subequations}
\begin{equation}\label{dual-spectrum1}
(E_l f)(\rho_0+\lambda):=E_{\lambda ,l} f(\rho_0+\lambda) \qquad (\lambda\in\Lambda)
\end{equation}
(on the domain of $f \in \ell^2(\rho_0+\Lambda,\Delta)$ for which $\langle E_l f,E_l f\rangle_{\Delta}<\infty$),
where
\begin{equation}\label{dual-spectrum2}
E_{\lambda ,l} :=q^{-\lambda_1-\lambda_2\cdots-\lambda_{l-1}}(q^{-\lambda_l}-1)+
t_0^2 q^{-\lambda_1-\lambda_2\cdots-\lambda_{n-1}}(q^{\lambda_n}-1)\delta_{n-l} .
\end{equation}
The $q$-difference operators $\hat{H}_l$  \eqref{dual-int} constitute nonnegative unbounded self-adjoint operators with purely discrete spectra in
$L^2(\mathbb{A},\hat{\Delta}\text{d}\xi)$ that are simultaneously diagonalized by the (inverse) deformed hyperoctahedral $q$-Whittaker transform $\boldsymbol{F}$ \eqref{ft1}, \eqref{ft2}:
\begin{equation}
\hat{H}_l=\boldsymbol{F}  \circ {{E}}_l \circ\boldsymbol{F}^{-1},\qquad l=1,\ldots ,n.
 \end{equation}
 \end{subequations}
\end{theorem}

\begin{proof}
It suffices to verify that 
\begin{equation*}\label{devec}
\hat{H}_l p_\lambda=E_{\lambda,l}p_\lambda \quad(\lambda\in\Lambda,\ l=1,\ldots ,n).
\end{equation*}
This is achieved  by multiplying the $l$th eigenvalue equation in Eq.~(5.8) of
\cite{die:properties} by a scaling factor $t^{l (n-l)+l(l -1)/2}$ and performing the limit $t\to 0$. Indeed, since the Macdonald-Koornwinder polynomial $\mathbf{p}_\lambda$ converges  to the deformed hyperoctahedral $q$-Whittaker function $p_\lambda$, we see from the explicit formulas for the operators in question that the LHS of the cited eigenvalue equation converges in this limit manifestly to $\hat{H}_l p_\lambda$ (up to an overall factor $t_0^l$). Hence, the  RHS  must also have a finite limit for $t\to 0$, which confirms that $p_\lambda$ is  an eigenfunction of $\hat{H}_l $ (using again that $\mathbf{p}_\lambda\stackrel{t\to 0}{\longrightarrow} p_\lambda$).  For $l >1$ it is not obvious from \cite[Eq.~(5.5)]{die:properties}  that
 the (limiting) eigenvalue is indeed given by $E_{\lambda ,l}$ \eqref{dual-spectrum2}, but this can be deduced quite easily from the asymptotics of $m_\lambda$ and $\hat{H}_l m_\lambda$ 
at $\xi=-ci \rho$, $\rho :=(n,n-1,\ldots,2,1)$ for $c\to +\infty$. Indeed,
one readily computes that for $c\to +\infty$:
$m_\lambda=e^{\langle \lambda ,\rho\rangle c}(1+o(1))$ and
$\hat{H}_l m_\lambda =E_{\lambda ,l}e^{\langle \lambda ,\rho\rangle c}(1+o(1))$
(using the explicit formula for $\hat{H}_l$ and the asymptotics
$$  \frac{\prod_{0\leq r\leq 3} (1-\hat{t}_re^{i\epsilon \xi_j}) }{(1-e^{2i\epsilon \xi_j})  (1-q e^{2i\epsilon \xi_j}) }
\stackrel{c\to +\infty}{\longrightarrow} \begin{cases} t_0^2 &\text{if}\ \epsilon =1 \\  1&\text{if}\ \epsilon =-1
\end{cases}\qquad (1\leq j\leq n)
$$
 and
$$
(1-q^a e^{i\epsilon (\xi_j\pm \xi_k)})^{-1}\ \stackrel{c\to +\infty}{\longrightarrow}\begin{cases}
0 &\text{if}\ \epsilon =1 \\ 1 &\text{if}\ \epsilon =-1
\end{cases}
\qquad ( 1\leq j<k\leq n) ,
$$
where $a\in \{ 1,0,-1\} $).
But then also
$p_\lambda=e^{\langle \lambda ,\rho\rangle c}(1+o(1))$ and
$\hat{H}_l p_\lambda =E_{\lambda ,l}e^{\langle \lambda ,\rho\rangle c}(1+o(1))$ for $c\to +\infty$ by
 the triangularity \eqref{qW1} and the property that
$\langle \mu,\rho\rangle<\langle \lambda,\rho\rangle$ if $\mu<\lambda$. The upshot is that the eigenvalue of $\hat{H}_l$ on the eigenpolynomial $p_\lambda$ must be equal to $E_{\lambda,l}$.
\end{proof}

The $q$-difference operators  $\hat{H}_l$ \eqref{dual-int} commute in the space of $W$-invariant trigonometric polynomials on $\mathbb{T}$.
It is clear from Theorem \ref{dual:thm} that this commutativity extends in the Hilbert space in the resolvent sense: for $$z_l \not\in\sigma (\hat{H}_l):=\{ E_{\lambda ,l}\mid\lambda\in\Lambda\}\subset [0,+\infty) \qquad
(l=1,\ldots ,n) $$
the resolvents $(\hat{H}_1-z_1\text{I})^{-1},\ldots ,(\hat{H}_n-z_n\text{I})^{-1}$ of the unbounded operators 
$\hat{H}_1,\ldots,\hat{H}_n$ mutually commute as bounded operators in $L^2(\mathbb{A},\hat{\Delta}\text{d}\xi)$.

Theorem \ref{dual:thm} and Section \ref{sec5}  lift the bispectral duality of $H$ \eqref{Haction} and $\hat{H}$  \eqref{hd1},\eqref{hd2}  to the complete systems of commuting quantum integrals.
The bispectral dual integrable system $\hat{H}_1,\ldots ,\hat{H}_n$ associated with our difference Toda chain can actually be identified as the strong-coupling limit ($t=q^g$, $g\to +\infty$) of a trigonometric Ruijsenaars-type
difference Calogero-Moser system with hyperoctahedral symmetry \cite{die:difference}.
Analogous bispectral dual systems were linked previously to the open quantum Toda chain and
Ruijsenaars' open difference Toda chain.
Specifically, the open quantum Toda chain and the strong-coupling limit of Ruijsenaars' rational difference
Calogero-Moser system turn out to be bispectral duals of each other
\cite{bab:equations,hal-rui:kernel,skl:bispectrality,koz:aspects}, and
 the same holds true for  Ruijsenaars' open difference Toda chain and the
$t=0$ trigonometric/hyperbolic Ruijsenaars-Macdonald operators \cite{ger-leb-obl:q-deformed,hal-rui:kernel,bor-cor:macdonald}.
Dualities of this type were actually first established for the corresponding particle systems within the realms of
 classical mechanics: the action-angle transforms linearizing the open Toda chain and the strong-coupling limit of the rational Ruijsenaars-Schneider system are the inverses of each other and the same holds true for
the action-angle transforms for Ruijsenaars' open relativistic Toda chain and the strong-coupling limit of the hyperbolic Ruijsenaars-Schneider system \cite{rui:relativistic,feh:action}.

\section{Parameter reductions}\label{sec7}
As already anticipated at the end of Section \ref{sec2},
for $\hat{t}_2=-\hat{t}_3=q^{1/2}$ and $\hat{t_0}=-\hat{t}_1\to 1$ (so $t_0=-t_1\to 1$ and $t_2=-t_3\to q^{1/2}$) the difference Toda Hamiltonian $H$ \eqref{Haction}  and the deformed hyperoctahedral $q$-Whittaker functions $p_\lambda (\xi)$, $\lambda\in\Lambda$ degenerate to a difference Toda Hamiltonian and $q$-Whittaker functions of type $D_n$ \cite{sur:discrete,kuz-tsy:quantum,eti:whittaker,sev:quantum,che:whittaker}. Even though formally these limiting values of the parameters do not respect our restriction that
$\hat{t}_r\in (-1,1)\setminus \{ 0\}$ (for $r=0,\ldots ,3$), it is readily inferred from the formulas that the results of Sections \ref{sec3}--\ref{sec7} nevertheless remain valid at this specialization of the parameters.

In this section we are concerned with the behavior for $\hat{t}_0\to 0$.  In this limit, the difference Toda chain  turns out to be governed by a Hamiltonian of the form
\begin{align}\label{H0}
\text{H} =&  T_1+ \sum_{j=2}^{n}  (1-q^{x_{j-1}-x_{j}})T_j +\sum_{j=1}^{n-1} (1-q^{x_j-x_{j+1}})T_j^{-1}\\
&  +\Bigl(\prod_{1\leq r<s\leq 3} (1-\hat{t}_r\hat{t}_sq^{x_n-1})\Bigr)  (1-q^{x_n})T_{n}^{-1} \nonumber \\
&+(\hat{t}_1+\hat{t}_2+\hat{t}_3)q^{x_n}+\hat{t}_1\hat{t}_2\hat{t}_3
q^{2x_n}(q^{x_{n-1}-x_n}+q^{-x_n-1}-1-q^{-1})  .\nonumber
\end{align}
When $\hat{t}_3=0$, the Hamiltonian in question constitutes a Ruijsenaars-type difference counterpart
of the quantum Toda chain
with one-sided boundary potentials of Morse type \cite{skl:boundary,ino:finite}. If in addition $\hat{t}_2=-1$,
then the difference Toda chain under consideration amounts to a quantization of a relativistic Toda chain with boundary potentials 
 introduced by Suris \cite{sur:discrete,kuz-tsy:quantum}.
For $\hat{t}_1=\hat{t}_2=\hat{t}_3=0$ and for $\hat{t}_1=-\hat{t}_2=q^{1/2}$ with $\hat{t}_3=-1$, we recover in turn hyperoctahedral difference Toda chains of type $B_n$ and $C_n$ that are diagonalized by
 $q$-Whittaker functions of  type $C_n$ and $B_n$, respectively \cite{eti:whittaker,sev:quantum,che:whittaker}.
Again, even though formally none of these specializations respect our restriction that
$\hat{t}_r\in (-1,1)\setminus \{ 0\}$ (for $r=1,2,3$), it is clear that the formulas below in fact do remain valid.

\subsection{Deformed hyperoctahedral $q$-Whittaker function}
For $\hat{t}_0\to 0$,  the deformed hyperoctahedral $q$-Whittaker functions $p_\lambda (\xi)$ \eqref{qW1}, \eqref{qW2} degenerate into a three-parameter  family of orthogonal polynomials
$\text{p}_\lambda(\xi)$, $\lambda\in\Lambda$
associated with the weight function
\begin{equation*}
\hat{\Delta} (\xi )=\frac{1}{(2\pi)^n}
\prod_{1\leq j<k\leq n}  \left| (e^{i(\xi_j+\xi_k)},e^{i(\xi_j-\xi_k)})_\infty \right|^2
\prod_{1\leq j\leq n}\left|
\frac{(e^{2i\xi_j})_\infty}{\prod_{1\leq r\leq 3} (\hat{t}_r e^{i\xi_j})_\infty}   \right|^2 .
\end{equation*}
The orthogonality relations for these polynomials read (cf. Eq. \eqref{qw-orthogonality})
\begin{equation}
\int_{\mathbb{A}} \text{p}_\lambda (\xi )\overline{ \text{p}_\mu (\xi)}\, \hat{\Delta}(\xi)\text{d}\xi =
\begin{cases}
{\Delta}_\lambda^{-1} &\text{if}\ \lambda =\mu ,\\
0&\text{otherwise},
\end{cases}
\end{equation}
where
\begin{equation*}
{\Delta}_\lambda= \frac{{\Delta}_0}{(q)_{\lambda_n}\prod_{1\leq r<s\leq 3} (\hat{t}_r\hat{t}_s)_{\lambda_n}}
\prod_{1\leq j<n}\frac{1}{(q)_{\lambda_j-\lambda_{j+1}}} 
\end{equation*}
with
\begin{equation*}
{\Delta}_0=(q)_\infty \prod_{1\leq r<s\leq 3} (\hat{t}_r\hat{t}_s)_\infty  .
\end{equation*}
For $n=1$, the limit $p_\lambda\stackrel{\hat{t}_0\to 0}{\longrightarrow} \text{p}_\lambda$ amounts to a well-known reduction from the Askey-Wilson polynomials to the continuous dual $q$-Hahn polynomials \cite{koe-les-swa:hypergeometric}.

\subsection{Hamiltonian}
The difference Toda eigenvalue equation
$H\psi_\xi =\hat{E}(\xi)\psi_\xi$ becomes in the limit $\hat{t}_0\to 0$ of the form $\text{H}\phi_\xi =\hat{E}(\xi)\phi_\xi$
with $\phi_\xi:\Lambda\to\mathbb{C}$ given by $\phi_\xi(\lambda)=\text{p}_\lambda(\xi)$ ($\xi\in\mathbb{A}$, $\lambda\in\Lambda)$, where $\text{H}$ \eqref{H0} acts on $f:\Lambda\to \mathbb{C}$ via
\begin{equation}\label{Haction0}
(\text{H} f)(\lambda)=
\sum_{\substack{1\leq j\leq n\\ \lambda +e_j\in\Lambda}} v_j^+(\lambda) f (\lambda +e_j)+ \sum_{\substack{1\leq j\leq n\\ \lambda -e_j\in\Lambda}}v_j^-(\lambda)f (\lambda -e_j) 
+u(\lambda )f (\lambda) ,\end{equation} 
with
\begin{align*}
v_j^+(\lambda) =&
(1-q^{\lambda_{j-1}-\lambda_j}), \\
v_j^-(\lambda) =&
(1-q^{\lambda_{j}-\lambda_{j+1}}) \Bigl((1-q^{\lambda_n})\prod_{1\leq r<s\leq 3} (1-\hat{t}_r\hat{t}_sq^{\lambda_n-1}) \Bigr)^{\delta_{n-j}} ,\\
u(\lambda )=& (\hat{t}_1+\hat{t}_2+\hat{t}_3)q^{\lambda_n}+\hat{t}_1\hat{t}_2\hat{t}_3
q^{2\lambda_n}(q^{\lambda_{n-1}-\lambda_n}+q^{-\lambda_n-1}-1-q^{-1}) 
\end{align*}
(subject to the convention that $\lambda_0=+\infty$ and $\lambda_{n+1}=-\infty$). 

\subsection{Diagonalization and integrability}
Let $\mathbf{F}: \ell^2(\Lambda,{\Delta})\to L^2(\mathbb{A},\hat{\Delta}\text{d}\xi)$
denote the ($\hat{t}_0\to 0$ degenerate) Hilbert space isomorphism determined by the orthogonal basis $\text{p}_\lambda$, $\lambda\in\Lambda$:
\begin{subequations}
\begin{equation}\label{ft10}
(\mathbf{F}{f})(\xi)= \langle {f},\phi_\xi \rangle_{{\Delta}}=\sum_{\lambda\in\Lambda}f(\lambda)
\overline{\phi_\xi (\lambda)}{\Delta}_\lambda
\end{equation}
(${f}\in \ell^2(\Lambda,{\Delta})$) 
with
\begin{equation}\label{ft20}
(\mathbf{F}^{-1}\hat{f})(\lambda) = \langle \hat{f},\overline{\phi ( \lambda)}\rangle_{\hat{\Delta}}=
\int_A \hat{f} (\xi) \phi_\xi (\lambda)\hat{\Delta}(\xi)\text{d}\xi
\end{equation}
\end{subequations}
($\hat{f}\in L^2(\mathbb{A},\hat{\Delta} \text{d}\xi)$), and let $\hat{{E}}_l:L^2(\mathbb{A},\hat{\Delta} \text{d}\xi)\to L^2(\mathbb{A},\hat{\Delta} \text{d}\xi)$ ($l=1,\ldots ,n$) be the multiplication operators defined in accordance with Section \ref{sec5}.

The commuting bounded self-adjoint operators 
$\text{H}_1,\ldots,\text{H}_n$ (with absolutely continuous spectra) in $\ell^2(\Lambda,{\Delta})$ given by
\begin{equation}
\text{H}_l=\mathbf{F}^{-1}  \circ \hat{{E}}_l \circ\mathbf{F},\qquad l=1,\ldots ,n,
 \end{equation}
constitute a complete system of quantum integrals for
the difference Toda Hamiltonian $\text{H}_1=\text{H}$ \eqref{Haction0}.

 \subsection{Bispectral dual system}
 Let  $\hat{\text{H}}_1,\ldots ,\hat{\text{H}}_n$ denote the commuting $q$-difference operators
 in Eq.  \eqref{dual-int} with $\hat{t}_0=0$ and let $\text{E}_1,\ldots,\text{E}_n$ be
 the self-adjoint multiplication operators in $ \ell^2(\Lambda,\Delta)$ given by (cf. Eqs. \eqref{dual-spectrum1}, \eqref{dual-spectrum2})
 \begin{subequations}
\begin{equation}
(\text{E}_l f)(\lambda)=\text{E}_{\lambda ,l} f(\lambda) \qquad (\lambda\in\Lambda, \, l=1,\ldots ,n)
\end{equation}
(on the domain of $f \in \ell^2(\Lambda,\Delta)$ for which $\langle \text{E}_l f,\text{E}_l f\rangle_{\Delta}<\infty$),
with 
\begin{equation}
\text{E}_{\lambda ,l} =q^{-\lambda_1-\lambda_2\cdots-\lambda_{l-1}}(q^{-\lambda_l}-1) .
\end{equation}
\end{subequations}
Then one has that
\begin{equation}
\hat{\text{H}}_l=\mathbf{F}  \circ {\text{E}}_l \circ\mathbf{F}^{-1},\qquad l=1,\ldots ,n,
 \end{equation}
 i.e. the $q$-difference operators constitute nonnegative unbounded self-adjoint operators with purely discrete spectra in
$L^2(\mathbb{A},\hat{\Delta}\text{d}\xi)$ that are simultaneously diagonalized by the three-parameter (inverse) deformed hyperoctahedral $q$-Whittaker transform $\mathbf{F}$ \eqref{ft10}, \eqref{ft20}.

\section{Scattering}\label{sec8}
In Ref. \cite{die:scattering} the scattering operator for a wide class of quantum lattice models was determined by stationary-phase methods originating from Ref. \cite{rui:factorized}. It follows from the diagonalization in Theorem \ref{diagonal:thm} that our difference Toda chains fit within this class of lattice models. Indeed,  the deformed hyperoctahedral $q$-Whittaker functions $p_\lambda$, $\lambda\in\Lambda$ belong to the family of orthogonal polynomials defined in
 \cite[\text{Sec.}~2]{die:scattering}, since the orthogonality weight function $\hat{\Delta}(\xi)$ \eqref{plancherel} 
 is of the indicated form (with $R=BC_n$) and moreover meets the demanded analyticity requirements.
We will close by briefly indicating how the general scattering results from Ref. \cite[\text{Sec.}~4.2]{die:scattering} specialize in the present difference Toda setting.

Let $\mathcal{H}_0$ be the self-adjoint discrete Laplacian in $\ell^2(\Lambda)$ of the form
\begin{equation*}
(\mathcal{H}_0 f)(\lambda)
:=
\sum_{\substack{1\leq j \leq n\\ \lambda+e_j\in\Lambda}} f(\lambda+e_j)
+\sum_{\substack{1\leq j \leq n\\ \lambda-e_j\in\Lambda}}   f(\lambda-e_j)
\qquad (f\in\ell^2(\Lambda)),
\end{equation*}
and let $\mathcal{H}$ denote the pushforward
\begin{equation}\label{pfH}
\mathcal{H}:= \boldsymbol{\Delta}^{1/2} H \boldsymbol{\Delta}^{-1/2}
\end{equation}
of the difference Toda Hamiltonian $H$ \eqref{Haction} onto the Hilbert space $\ell^2(\Lambda)$ via the Hilbert space isomorphism
$\boldsymbol{\Delta}^{1/2}:\ell^2(\rho_0+\Lambda,\Delta)\to \ell^2(\Lambda)$  given by
\begin{equation}\label{embed1}
(\boldsymbol{\Delta}^{1/2}f)(\lambda):=\Delta^{1/2}_\lambda f(\rho_0+\lambda)
\qquad (f\in \ell^2(\rho_0+\Lambda,\Delta) )
\end{equation}
(where $\boldsymbol{\Delta}^{-1/2}:=(\boldsymbol{\Delta}^{1/2})^{-1}$). Clearly, one has by Theorem \ref{diagonal:thm} that
\begin{equation}\label{F}
\mathcal{H}=\boldsymbol{\mathcal{F}}^{-1} \hat{E} \boldsymbol{\mathcal{F}}
\quad
\text{with} 
\quad
\boldsymbol{\mathcal{F}}:= \boldsymbol{\hat{\Delta}}^{1/2} \boldsymbol{F} \boldsymbol{\Delta}^{-1/2} ,
\end{equation}
where $\boldsymbol{\hat{\Delta}}^{1/2}:L^2(\mathbb{A},\hat{\Delta}\text{d}\xi)\to L^2(\mathbb{A})$ denotes the Hilbert space isomorphism given by
\begin{equation}
(\boldsymbol{\hat{\Delta}}^{1/2}\hat{f})(\xi):= \hat{\Delta}^{1/2}(\xi)\hat{f}(\xi)
\qquad (\hat{f}\in L^2(\mathbb{A},\hat{\Delta}\text{d}\xi))
\end{equation}
(and $\hat{E}$ \eqref{E} is now regarded as a self-adjoint bounded multiplication operator in $L^2(\mathbb{A})$).
Moreover, it is elementary that the spectral decomposition of the discrete Laplacian $\mathcal{H}_0$ is given by
$$
\mathcal{H}_0=\boldsymbol{\mathcal{F}}_0^{-1} \hat{E} \boldsymbol{\mathcal{F}}_0,
$$
where $\boldsymbol{\mathcal{F}}_0:\ell^2(\Lambda)\to L^2(\mathbb{A})$ denotes the Fourier isomorphism
\begin{subequations}
\begin{equation}
(\boldsymbol{\mathcal{F}}_0{f})(\xi):=\sum_{\lambda\in\Lambda}f(\lambda)
\overline{\chi_\xi (\lambda)}
\end{equation}
(${f}\in \ell^2(\Lambda)$) 
with the inversion formula
\begin{equation}
(\boldsymbol{\mathcal{F}}_0^{-1}\hat{f})(\lambda) =
\int_{\mathbb{A}} \hat{f} (\xi) \chi_\xi (\lambda)\text{d}\xi
\end{equation}
\end{subequations}
($\hat{f}\in L^2(\mathbb{A})$). Here we have employed the anti-invariant Fourier kernel
$$
 \chi_\xi (\lambda):=\frac{1}{(2\pi )^{n/2}\, i^{n^2}} \sum_{w\in W} \text{sign}(w) e^{i\langle w(\rho +\lambda ) ,\xi\rangle} ,
$$
with $\text{sign}(w)=\epsilon_1\cdots\epsilon_n\text{sign}(\sigma)$ for $w=(\sigma,\epsilon)\in W=S_n\ltimes \{ 1,-1\}^n$ and 
$\rho= (n,n-1,\ldots,2,1)$. Notice that $\mathcal{F}_0$ is recovered from $\mathcal{F}$ in the limit $q\to 0$,
$\hat{t}_r\to 0$ ($r=0,\ldots ,3$).

The scattering operator describing the large-times asymptotics of the difference Toda dynamics
$e^{i\mathcal{H}t}$ relative to the Laplacian's reference dynamics $e^{i\mathcal{H}_0t}$ turns out to be governed by an
$n$-particle scattering matrix $\hat{ {\mathcal S}} (\xi)$ that factorizes in two-particle pair matrices and one-particle boundary matrices:
\begin{subequations}
\begin{equation}
\hat{ {\mathcal S}} (\xi)
 := \prod_{1\leq j<k\leq n} s(\xi_j-\xi_k)s(\xi_j+\xi_k)\prod_{1\leq j\leq n} s_0(\xi_j) ,
\end{equation}
with
\begin{equation}
s(x):=\frac{(qe^{ix})_\infty }{(qe^{-ix})_\infty }\quad\text{and}\quad  
s_0(x):=
\frac{(qe^{2ix})_\infty }{(qe^{-2ix})_\infty }\prod_{0\leq r\leq 3}\frac{(\hat{t}_re^{-ix})_\infty }{(\hat{t}_re^{ix})_\infty }.
\end{equation}
\end{subequations}
To make the latter statement precise, let us denote by $C_0(\mathbb{A}_{\text{reg}})$ the dense subspace of $L^2(\mathbb{A})$ consisting of smooth test functions with compact support in
 the open dense subset $\mathbb{A}_{\text{reg}}\subset\mathbb{A}$ on which the components of the gradient
$$\nabla \hat{E}(\xi)=(-2\sin(\xi_1),\ldots,-2\sin(\xi_n)),\quad \xi\in\mathbb{A}$$ do not vanish and are all distinct in absolute value.
 We now define an unitary multiplication operator $ \hat{\mathcal S} : L^2(\mathbb{A},\text{d}\xi)\to  L^2(\mathbb{A},\text{d}\xi)$ via its restriction to $C_0(\mathbb{A}_{\text{reg}})$ as follows:
 \begin{equation}
 ( \hat{\mathcal S}\hat{f})(\xi):=  \hat{\mathcal S}(w_\xi \xi )\hat{f}(\xi)\qquad
 (\hat{f}\in C_0(\mathbb{A}_{\text{reg}}),
  \end{equation}
 where $w_\xi\in W$ for $\xi\in \mathbb{A}_{\text{reg}}$ is such that the components of $w_\xi \nabla \hat{E}(\xi)$
 are all positive and reordered from large to small.

Theorem~4.2 and Corollary~4.3 of
Ref. \cite{die:scattering} then provide the following explicit formulas for the wave operators and scattering operator of our difference Toda chain.
\begin{theorem}[Wave and Scattering Operators]\label{scattering:thm}
  The operator limits
\begin{equation*}
\Omega^{\pm} :=s-\lim_{t\to \pm \infty}  e^{i t  \mathcal{H}}e^{-it \mathcal{H}_{0}}
\end{equation*}
converge in the strong $\ell^2(\Lambda)$-norm topology and the corresponding wave operators $\Omega^\pm$ intertwining the difference Toda
 dynamics $e^{i\mathcal{H}t}$  with the discrete Laplacian's dynamics $e^{i\mathcal{H}_0t}$
are given by unitary operators in $\ell^2(\Lambda )$ of the form
\begin{equation*}
\Omega^\pm = \boldsymbol{\mathcal{F}}^{-1} \circ \hat{\mathcal S}^{\mp 1/2}  \circ \boldsymbol{\mathcal{F}}_0,
\end{equation*}
where the branches of the square roots are to be chosen such that
\begin{equation*}
s(x)^{1/2}=\frac{(qe^{ix})_\infty }{|(qe^{ix})_\infty | }
\quad\text{and}\quad
 s_0(x)^{1/2}=\frac{(qe^{2ix})_\infty }{|(qe^{2ix})_\infty | }\prod_{0\leq r\leq 3}\frac{|(\hat{t}_re^{ix})_\infty |}{(\hat{t}_re^{ix})_\infty }  .
 \end{equation*}

Hence, the scattering operator relating the large-times asymptotics of the
difference Toda
 dynamics $e^{i\mathcal{H}t}$  for $t\to - \infty$ and $t\to +\infty$ is given by the unitary operator
\begin{equation*}
\mathcal{S}:=(\Omega^+)^{-1} \Omega^- =  \boldsymbol{\mathcal{F}}_0^{-1}  \circ \hat{\mathcal S} \circ  \boldsymbol{\mathcal{F}}_0 .
\end{equation*}
\end{theorem}
 The  degenerate case of  the difference Toda chain discussed in Section \ref{sec7} is also covered by Theorem \ref{scattering:thm}, upon setting
$\rho_0$ equal to the nulvector in Eq. \eqref{embed1},
 replacing $H$ \eqref{Haction} by $\text{H}$ \eqref{Haction0} in $\mathcal{H}$ \eqref{pfH} 
 and $\boldsymbol{F}$ \eqref{ft1}, \eqref{ft2} by $\mathbf{F}$ \eqref{ft10}, \eqref{ft20} in $\boldsymbol{\mathcal{F}}$ \eqref{F}, and
 substituting  $\hat{t}_0=0$ overall.

\appendix

\section{Macdonald-Koornwinder polynomials}\label{appA}
This appendix collects some key properties of the Macdonald-Koornwinder multivariate Askey-Wilson polynomials \cite{koo:askey-wilson,die:properties,mac:affine}. In the case of one variable ($n=1$), the properties below specialize to well-known formulas for the
Askey-Wilson polynomials (see e.g. \cite{koe-les-swa:hypergeometric}).

The Macdonald-Koornwinder polynomials  $\mathbf{p}_{\lambda} (\xi)$ ($\lambda\in\Lambda$, $\xi \in \mathbb{T}$)  are defined as polynomials of the type in Eqs. \eqref{qW1}, \eqref{qW2}, \eqref{ip}  associated with the weight function \cite[\text{Sec.}~5]{koo:askey-wilson},  \cite[\text{Ch.}~5.3]{mac:affine}:
\begin{equation*}\label{MKweight}
\boldsymbol{\hat{\Delta}} (\xi )=\frac{1}{(2\pi)^n}
\prod_{1\leq j\leq n}\Bigl|
\frac{(e^{2i\xi_j})_\infty}{\prod_{0\leq r\leq 3} (\hat{t}_re^{i\xi_j})_\infty}  \Bigr|^2 
\prod_{1\leq j<k\leq n}\Bigl| 
\frac{(e^{i(\xi_j+\xi_k)},e^{i(\xi_j-\xi_k)})_\infty}{(te^{i(\xi_j+\xi_k)},te^{i(\xi_j-\xi_k)})_\infty}\Bigr|^2
,
\end{equation*}
with $q\in (0,1)$ and $t,\hat{t}_r\in (-1,1)\setminus \{0\}$ ($r=0,\ldots ,3)$.
For $t\to 0$ this weight function passes into that of Eq. \eqref{plancherel}, whence the polynomials in question degenerate in this limit continuously to the deformed hyperoctahedral $q$-Whittaker functions of Section \ref{sec3}.
Notice in this respect that for $x\in\mathbb{R}$ and $| t | <\varepsilon$ ($<1$) quotients of the form
$(e^{ix})_\infty/ (te^{ix})_\infty$ remain bounded in absolute value by $(-1)_\infty/(\varepsilon)_\infty$, so  we may interchange limits and integration for $t\to 0$ when integrating trigonometric polynomials against the Macdonald-Koornwinder weight function $\boldsymbol{\hat{\Delta}} (\xi )$ over the bounded alcove $\mathbb{A}$ (by dominated convergence). 

The normalized Macdonald-Koornwinder polynomials
\begin{subequations}
\begin{equation}\label{n-pol1}
\mathbf{P}_\lambda (\xi ):=\mathbf{c}_\lambda \mathbf{p}_\lambda (\xi) \qquad (\lambda\in\Lambda_n),
\end{equation}
where
\begin{equation}\label{n-pol2}
\mathbf{c}_\lambda :=
\prod_{1\leq j\leq n} \frac{(\tau_j^2)_{2\lambda_j}}{\prod_{0\leq r\leq3} (t_r\tau_j)_{\lambda_j}}
\prod_{1\leq j<k\leq n}
\frac{(\tau_j\tau_k)_{\lambda_j+\lambda_k}}{(t\tau_j\tau_k)_{\lambda_j+\lambda_k}}
\frac{(\tau_j\tau_k^{-1})_{\lambda_j-\lambda_k}}{(t\tau_j\tau_k^{-1})_{\lambda_j-\lambda_k}}
\end{equation}
\end{subequations}
with $\tau_j:=t^{n-j}t_0$ ($j=1,\ldots ,n$) and $t_r$ ($r=0,\ldots ,3$) given by Eq. \eqref{tp1}, satisfy the following orthogonality relations
\cite[\text{Sec.}~5]{koo:askey-wilson}, \cite[\text{Sec.}~7]{die:properties}, \cite[\text{Ch.}~5.3]{mac:affine}:
\begin{subequations}
\begin{equation}\label{mk-ortho-a}
\int_{\mathbb{A}} \mathbf{P}_\lambda (\xi) \overline{ \mathbf{P}_\mu (\xi)}
\boldsymbol{\hat{\Delta}}(\xi) \text{d}\xi =
\begin{cases}
\boldsymbol{\Delta}_\lambda^{-1} &\text{if}\ \lambda = \mu ,\\
0 &\text{otherwise},
\end{cases}
\end{equation}
with
\begin{align}
\boldsymbol{\Delta}_\lambda  :=&\boldsymbol{\Delta}_0 
 \prod_{1\leq j\leq n}\Biggl( \frac{1-\tau_j^2q^{2\lambda_j}}{1-\tau_j^2} 
\prod_{0\leq r\leq3} 
\frac{(t_r\tau_j)_{\lambda_j}} {(qt_r^{-1}\tau_j)_{\lambda_j}}   \Biggr)  \\
\times &\prod_{1\leq j<k\leq n}\frac{1-\tau_j\tau_kq^{\lambda_j+\lambda_k}}{1-\tau_j\tau_k}
\frac{(t\tau_j\tau_k)_{\lambda_j+\lambda_k}}{(qt^{-1}\tau_j\tau_k)_{\lambda_j+\lambda_k}}
\frac{1-\tau_j\tau_k^{-1}q^{\lambda_j-\lambda_k}}{1-\tau_j\tau_k^{-1}}
\frac{(t\tau_j\tau_k^{-1})_{\lambda_j-\lambda_k}}{(qt^{-1}\tau_j\tau_k^{-1})_{\lambda_j-\lambda_k}} 
\nonumber
\end{align}
and
\begin{equation}\label{mk-ortho-d}
\boldsymbol{\Delta}_0:=
\prod_{1\leq j\leq n} \frac{(q,t^j)_\infty \prod_{0\leq r<s\leq 3} (\hat{t}_r\hat{t}_st^{n-j})_\infty}{(t,\hat{t}_0\hat{t}_1\hat{t}_2\hat{t}_3t^{2n-j-1})_\infty} .
\end{equation}
\end{subequations}

These orthogonal polynomials satisfy moreover a second-order $q$-difference equation
\cite[\text{Sec.}~5]{koo:askey-wilson}, \cite[\text{Ch.}~5.3, 4.4]{mac:affine}:
\begin{align}\label{qde}
&\mathbf{P}_\lambda (\xi ) \sum_{j=1}^n\bigl( q^{-1}\hat{t}_0 \hat{t}_1\hat{t}_2\hat{t}_3t^{2n-1-j}(q^{\lambda_j}-1)+t^{j-1}(q^{-\lambda_j}-1)\bigr)= \\
&\sum_{1\leq j\leq n}
\hat{V}_j(\xi) \bigl(  \mathbf{P}_{\lambda }(\xi-i\log (q) e_j )-\mathbf{P}_\lambda(\xi ) \bigr)
+
\hat{V}_j(-\xi) \bigl(  \mathbf{P}_{\lambda }(\xi+i\log (q) e_j)-\mathbf{P}_\lambda(\xi ) \bigr),\nonumber
\end{align}
with
\begin{equation*}
\hat{V}_{ j}(\xi):=
\frac{\prod_{0\leq r\leq 3} (1-\hat{t}_re^{i\xi_j}) }{(1-e^{2i\xi_j})  (1-q e^{2i\xi_j}) }
\prod_{\substack{1\leq k\leq n\\k\neq j}} \frac{1-te^{i(\xi_j+\xi_k)}}{1-e^{i(\xi_j+\xi_k)}}
\frac{1-te^{i(\xi_j-\xi_k)}}{1-e^{i(\xi_j-\xi_k)}} ,
\end{equation*}
and a Pieri-type recurrence formula  \cite[\text{Sec.}~6]{die:properties}, \cite[\text{Ch.}~5.3, 4.4]{mac:affine}:
\begin{eqnarray}\label{pieri}
\lefteqn{\mathbf{P}_\lambda (\xi ) \sum_{j=1}^n (2\cos (\xi_j)-\hat{\tau}_j-\hat{\tau}_j^{-1})=} &&\\
&&\sum_{\substack{1\leq j\leq n\\ \lambda +e_j\in\Lambda}}
V_j^+(\lambda) \left( \hat{\tau}_j \mathbf{P}_{\lambda +e_j}(\xi )-\mathbf{P}_\lambda(\xi ) \right)
+\sum_{\substack{1\leq j\leq n\\ \lambda -e_j\in\Lambda}}
V_j^-(\lambda) \left(  \hat{\tau}_j^{-1} \mathbf{P}_{\lambda -e_j}(\xi )-\mathbf{P}_\lambda(\xi ) \right),\nonumber
\end{eqnarray}
with $\hat{\tau}_j:=t^{n-j}\hat{t}_0$ ($j=1,\ldots ,n$) and
\begin{align*}
V_j^+(\lambda) &:= \frac{\hat{\tau}_1^{-1}\prod_{0\leq r\leq 3}(1-t_r\tau_jq^{\lambda_j})}{(1-\tau_j^2q^{2\lambda_j})(1-\tau_j^2q^{2\lambda_j+1})}
\prod_{\substack{1\leq k\leq n\\ k\neq j}}
\frac{1-t\tau_j\tau_kq^{\lambda_j+\lambda_k}}{1-\tau_j\tau_kq^{\lambda_j+\lambda_k}}
\frac{1-t\tau_j\tau_k^{-1}q^{\lambda_j-\lambda_k}}{1-\tau_j\tau_k^{-1}q^{\lambda_j-\lambda_k}}
,\\
V_j^-(\lambda) &:= 
 \frac{\hat{\tau}_1\prod_{0\leq r\leq 3}(1-t_r^{-1}\tau_jq^{\lambda_j})}{(1-\tau_j^2q^{2\lambda_j})(1-\tau_j^2q^{2\lambda_j-1})}
\prod_{\substack{1\leq k\leq n\\k\neq j}}
\frac{1-t^{-1}\tau_j\tau_kq^{\lambda_j+\lambda_k}}{1-\tau_j\tau_kq^{\lambda_j+\lambda_k}}
\frac{1-t^{-1}\tau_j\tau_k^{-1}q^{\lambda_j-\lambda_k}}{1-\tau_j\tau_k^{-1}q^{\lambda_j-\lambda_k}} 
\end{align*}
(where the vectors $e_1,\ldots ,e_n$ refer to the standard unit basis of $\mathbb{R}^n$).

\bibliographystyle{amsplain}

\end{document}